\documentclass[12pt]{amsart}
\usepackage{amssymb}
\usepackage{latexsym}
\usepackage{amsfonts,euscript}
\usepackage{tabularx}
\usepackage{array}
\usepackage{float}
\usepackage{booktabs}
\usepackage{ctable}
\usepackage[pdftex]{hyperref} 
\usepackage{subfigure}
\hypersetup{ 
colorlinks,%
citecolor=blue,%
filecolor=black,%
linkcolor=black,%
urlcolor=black 
} 

\usepackage{sidecap}
\usepackage{verbatim}

\usepackage{graphicx}

\newtheorem{theorem}{Theorem}[section]
\newtheorem{lemma}{Lemma}[section]
\newtheorem{proposition}{Proposition}[section]
\newtheorem{corollary}{Corollary}[section]

\theoremstyle{definition}

\numberwithin{equation}{section}

\newcommand{\D}{\displaystyle}

\begin{document}
\bibliographystyle{abbrv}
\title{Forward Hysteresis and Backward Bifurcation Caused by Culling in an Avian Influenza Model}


\author{Hayriye Gulbudak$^*$}
\address{Department of Mathematics,  
University of Florida,
358 Little Hall,
PO Box 118105,
Gainesville, FL 32611--8105 }
\email{hgulbudak@ufl.edu}
\author{Maia Martcheva}
\address{Department of Mathematics, 
University of Florida,
358 Little Hall,
PO Box 118105,
Gainesville, FL 32611--8105 }
\email{maia@ufl.edu}

\thanks{$^*$author for correspondence}

\begin{abstract}

The emerging threat of a human pandemic caused by the H5N1 avian influenza virus strain magnifies the need for controlling the incidence of H5N1 infection in domestic bird populations. Culling is one of the most widely used control measures and has proved effective for isolated outbreaks.  However, the socio-economic impacts of mass culling, in the face of a disease which has become endemic in many regions of the world, can affect the implementation and success of culling as a control measure. We use mathematical modeling to understand the dynamics of avian influenza under different culling approaches. We incorporate culling into an SI model by considering the per capita culling rates to be general functions of the number of infected birds. Complex dynamics of the system, such as backward bifurcation and forward hysteresis, along with bi-stability, are detected and analyzed for two distinct culling scenarios. In these cases, employing other control measures temporarily can drastically change the dynamics of the solutions to a more favorable outcome for disease control.  

\bigskip


\bigskip

\end{abstract}
\date{\today}
\maketitle
\pagestyle{myheadings}
\markboth{\sc }
{\sc Forward Hysteresis, Backward Bifurcation and Culling in  H5N1}

\baselineskip14pt

\section{Introduction}

 H5N1 (highly pathogenic avian influenza) has rapidly spread among wild and domestic bird populations in recent years.  With increasing frequency, the virus has shown the ability to infect mammalian species which are in close contact with infected birds \cite{PJG, FSPH, WHO3, CDC1, WHO2}.  Most notably, over 600 humans have contracted H5N1 since 1997 with a reported 60$\%$ mortality rate \cite{WHO1, CDC}. The most serious public health threat that H5N1 poses to humans is the potential appearance of an extremely virulent human-to-human transmittable strain of avian influenza \cite{Fleck, Fouchier, Rezza, ZSL, Wang1}. Reducing the probability of this occurring requires strong control measures. However, avian influenza is a complex disease, infecting multiple species of animals, which creates difficulties in tracking and controlling the disease. Hence, control has been directed at reducing incidence among poultry populations, since these are the main animal populations responsible for transmitting the disease to humans.  
 
   One of the main control measures applied to poultry is culling, i.e the targeted elimination of a portion of the poultry population in areas affected by avian influenza, to save the rest of the birds and reduce the possibility of further outbreaks. Wide-area (mass) culling has been successful for isolated outbreaks. However, mass elimination of poultry becomes too much of an economic burden in areas of wide-spread outbreaks and in countries dominated by smallholder farms. As the number of infected increases above a threshold level, it has been suggested that control measures be shifted from mass culling to a modified strategy. The modified strategy includes elimination of only infected flocks and high risk in-contact birds along with other control measures \cite{FAO0, FAO00}. In areas with backyard poultries, the selective culling of only infected and highly exposed flocks is often employed instead of mass culling, even for small outbreaks \cite{FAO0, FAO00, Modcull, WHO0}. 
   
   Mathematical modeling provides a way of understanding the complex epidemiology of avian influenza and can yield valuable insights on how different control strategies impact the disease dynamics. There have been several authors who have looked at culling in particular. Le Menach et al. analyzed a spatial farm-based model, which treats poultry farms as units, and found that an immediate depopulation of infected flocks following an accurate and quick diagnosis would have a greater impact than simply depopulating surrounding flocks \cite{culltrans}. Martcheva investigated the efficacy of culling in comparison with other control measures, and determined that culling without repopulation is the most effective control measure based upon sensitivity analysis \cite{Martcheva5}. Iwami et al. investigated a mathematical model for the spread of wild and mutant avian influenza, and explored the effectiveness of the prevention policies, namely elimination and quarantine policy \cite{Iwami1, Iwami2, Iwami3}. Shim and Galvani evaluate the effect of culling on the host-pathogen evolution \cite{Shim}.  Impulsive systems have also been considered for modeling culling, which will be discussed further in Section \ref{sec2}.

   Even though mathematical and statistical models have been focused on culling as a control strategy, almost no special attention has been put on how culling strategies differ from region to region as a result of socio-economic factors. In this article, we model avian influenza dynamics in domestic birds under the control measure of culling, giving special attention to these different culling strategies. We incorporate, into an SI model, various culling rates that are functions of the number of infected birds in the population. 
   
   For certain culling rates in our model we find complex bifurcations, namely backward bifurcation and forward hysteresis. In epidemiological models with backward bifurcation, the disease may persist even though the basic reproduction number, $\mathcal R_0$, is less then $1$. In this case, as $\mathcal R_0$ approaches unity from the left, there exist endemic equilibria in addition to the locally stable disease free equilibrium. This leads to bistable dynamics: If the initial number of infected individuals is small enough, the disease will die out; if, however, the disease level is above some threshold, then the disease will persist.  Backward bifurcations have been found and explored in several models from mathematical epidemiology \cite{Brauer, Li, Wang, Gumel, Martcheva, Zhou-SIRtreatment}.  On the other hand, forward hysteresis has rarely been detected or studied in epidemiological models. Forward hysteresis refers to the existence of multiple endemic equilibria and bistable dynamics when $\mathcal R_0>1$.  In this case, although the disease will always persist, there can be a dramatic difference in the asymptotic level of disease, depending on the initial conditions. Hu et al. recently studied an SIR model with saturating incidence and piecewise defined treatment, which they showed to exhibit forward hysteresis \cite{Hu}.
   
  Bistable dynamics can have important implications for control. During the onset of an outbreak, actions can be taken in order to ``drive the solution'' to the region of attraction corresponding to a low level equilibrium or disease free state. For example, a temporary reduction in transmission rate can produce this shift in asymptotic dynamics. Temporary control measures such as enhanced biosecurity and movement ban on poultry have the impact of reducing the transmission rate and, hence, may be important in ultimately bringing the disease under control.
   
   In section 2, we introduce our general model, derive the reproduction number $\mathcal R_0$ for the general model, analyze the general model, and describe the distinct culling functions that we will incorporate into the system. In section 3, we consider the case of mass culling as an example. In section 4 and 5, we give the motivation behind employment of modified culling and selective culling, incorporate the corresponding culling rates in the system, and analyze the resulting models. In section 6, we consider the implications of the bistable dynamics observed under selective and modified culling and present numerical simulations which highlight the impact of temporary control measures. In section 7, we conclude with a discussion about our results and their implicaations.       

\section{Modeling culling in avian influenza H5N1}\label{sec2}
  The first model of H5N1 influenza was introduced by Iwami et al. \cite{Iwami0,Iwami1,Iwami2}. The model does not explicitly take into account any control measures. In practice, culling infected and exposed poultry has been utilized in order to reduce the global spread of H5N1. We alter the basic avian flu model in birds to explicitly involve culling as a control measure.   Since culling is applied only after an outbreak has occurred, we assume that the culling rate depends on the number of infected individuals.
  The general model takes the form:
\begin{equation}
\label{model1}
{M:} \quad
\begin{cases}
\D\frac{d S}{dt} &\hspace{-4mm}= \Lambda -\displaystyle\frac{\beta I S}{N} - \mu
S-c_S \phi(I)S  \vspace{1.5mm},\\
\D\frac{dI}{dt} &\hspace{-4mm}=\displaystyle\frac{\beta I S}{N} -  (\mu +
\nu)I - c_I \psi(I) I\vspace{1.5mm}\\
\end{cases}
\end{equation}
with nonnegative initial conditions : $S(0), I(0) \geq 0.$  The state variables $S$ and $I$ represent the number of susceptible domestic birds and the number of infected domestic birds, respectively.  The total number of domestic birds is denoted by $N$, where $N=S+I$. For the parameters, $\Lambda$ is the recruitment rate of domestic birds, $\beta$ is the
transmission rate, $\mu$ is the mortality rate of domestic birds, $\nu$ is
the disease-induced death rate for birds, and $c_S$ and $c_I$ are the culling
constants for susceptible and infected birds respectively.  We assume susceptible domestic birds $S$ are culled at a rate 
$c_S \phi(I)$ and infected domestic birds $I$ are culled at a potentially higher
rate $c_I \psi(I)$.  The culling functions $\psi(I)$ and $\phi(I)$ are assumed to be non-negative continuous functions. 

Our choice of modeling culling as a continuous function dependent on the number of infected, $I$, reflects our aim to consider different culling strategies in which culling effort depends on $I$, i.e. increases, decreases or non-monotone with respect to $I$.   Impulsive systems can also been considered for modeling culling.  An advantage of modeling culling as a pulsed process is that culling does not occur continuously through time.  Terry \cite{terry} considered an impulsive system to model culling of crop pests, where the culling occurred as pulses applied at fixed times.  However, employment of culling at fixed times may not be realistic for avian influenza since it ignores the fact that culling occurs as a response to outbreaks.  Another possibility is state dependent impulsive models, which was considered for pulse vaccination in an SIR model \cite{pulseSIR}.  In this approach, impulsive culling would occur upon $I$ reaching a threshold value, but culling effort would not vary beyond this impulse switch and limited qualitative results can be obtained in such a model.  A limitation of our model is that culling occurs as a continuous, ongoing process, which may not be realistic.  There are advantages and disadvantages to each modeling approach, but considering continuous culling rate functions dependent on $I$, as in system \eqref{model1}, may help to better understand the dynamical consequences of the different culling strategies mentioned in the Introduction.

\begin{table}[h]
\caption{Definition of the variables  in the modeling 
framework} 
\label{ttable:variables}
\centering 
\begin{tabularx}{\textwidth}{>{} lX}
\toprule
Variable/Parameter  & Meaning \\ [0.5ex]
\toprule
\\
$S$ & Susceptible domestic birds  \\ [0.5ex]	
$I$& Birds infected with HPAI \\ [0.5ex]
$\Lambda$ & Birth/recruitment rate of domestic birds \\[0.5 ex] 
$\beta$ & Transmission rate of HPAI among domestic birds \\[0.5 ex]
$\phi(I)$ & Culling rate for susceptible poultry \\ [0.5 ex]
$\psi(I)$ & Culling rate for infected poultry  \\[0.5 ex] 
$c_S$ & Culling coefficient for susceptible poultry \\[0.5 ex]
$c_I$ & Culling coefficient for infected  poultry \\ [0.5 ex]
$\mu$ & Natural death rate of domestic birds \\[0.5 ex]
$\nu$ & HPAI-induced mortality rate for domestic birds \\[0.5 ex]
\bottomrule
\end{tabularx}
\end{table}

\subsection{Analysis of General Model}
The solutions of \eqref{model1} are non-negative for all time $t$. Moreover, there is a positively invariant compact set $$K=\left\{(S,I)\in  \mathbb{R}^2: S \geq 0,\ I \geq 0, S+I \leq \frac{\Lambda}{\mu} \right\}$$ in the non-negative quadrant of $\mathbb{R}^2$ which attracts all solutions of \eqref{model1}.  Indeed, by adding the equations in \eqref{model1}, we see that $N'\leq \Lambda -\mu N$.  Hence, for any solution $\left(S(t),I(t)\right)$,
\begin{align*}
0\leq \limsup_{t\rightarrow\infty}S(t),\limsup_{t\rightarrow\infty}I(t)\leq \limsup_{t\rightarrow\infty}N(t)\leq\frac{\Lambda}{\mu}.
\end{align*}

The system \eqref{model1} has a disease free equilibrium denoted by $\mathcal E_0$, where 
\begin{align}
\label{DFE0}
\mathcal E_0&=\left(\frac{\Lambda}{\mu},0 \right).
\end{align}
  In order to calculate the reproduction number $\mathcal R_0$, we find a threshold condition for the local stability of $\mathcal E_0$. By computing the Jacobian matrix evaluated at $\mathcal E_0$, one can derive the following formula for $\mathcal R_0$ and the following theorem:
\begin{align}
\label{RN}
\mathcal R_0&=\frac{\beta}{\mu+\nu+c_I\psi(0)}
\end{align}

\begin{theorem}
If $\mathcal R_0<1$, then the disease free equilibrium, $\mathcal E_0$, is locally asymptotically stable for the system \eqref{model1}.  If $\mathcal R_0>1$, then $\mathcal E_0$ is unstable.
\end{theorem}

Observe that $\mathcal R_0$ is a function of $c_I\psi(0)$.  If the per capita culling rate function $\psi(I)$ is zero at $I=0$, then the culling rate does not affect the value of $\mathcal R_0$.

Under a certain condition on the per capita culling rate $\psi(I)$, we can obtain the following global result when $\mathcal R_0<1$:
\begin{theorem}\label{GAS}
If $\mathcal R_0< 1$ and $\psi(I)$ satisfies the following condition:
\begin{align} 
\label{Cond1}
\psi(0)=\inf_{I \geq 0}\psi(I),
\end{align}

then $\mathcal E_0$ is globally asymptotically stable.
\end{theorem}
\begin{proof}
 By the last equation in \eqref{model1}, we have the following inequality:
\begin{align*}
I'&=\frac{\beta SI}{N}-(\mu+\nu+c_I\psi(I))I \\
&\leq \left[\beta-\left(\mu+\nu+c_I\inf_{I\geq 0}\psi(I)\right)\right]I.
\end{align*}
 
Hence $\lim_{t\rightarrow \infty} I = 0$ when $\mathcal R_0 < 1$ and (\ref{Cond1}) is satisfied.
\end{proof}

The direction of the transcritical bifurcation at $\mathcal R_0=1$ determines whether an endemic equilibrium $\mathcal E^*=(S^*,I^*)$ with low levels of $I^*$ exists for $\mathcal R_0>1$ or $\mathcal R_0<1$.  More precisely, the transcritical bifurcation is forward (backward) if there exists a positive equilibrium $\mathcal E^*$ when $R_0>1$ ($\mathcal R_0<1$), in which $I^*\rightarrow 0$ as $\mathcal R_0 \rightarrow 1$ from the right hand side (left hand side).

\begin{theorem}\label{fb}
Consider the system \eqref{model1}.  If
$$\psi'(0)>-\frac{(\mu+\nu+c_I\psi(0))(\mu+c_s\phi(0))}{c_I\Lambda},$$
then there is a forward bifurcation at $\mathcal R_0=1$.  However, if 
$$\psi'(0)<-\frac{(\mu+\nu+c_I\psi(0))(\mu+c_s\phi(0))}{c_I\Lambda},$$
then there is a backward bifurcation at $\mathcal R_0=1$.
\end{theorem}

Note that by the theorem above, if $\psi'(0)\geq 0$, then there is a forward bifurcation, in particular it excludes the presence of backward bifurcation.  In other words, for backward bifurcation to occur, it is necessary that $\psi'(0)$ is negative.

\begin{proof}
Consider the equations for endemic equilibria $(S^*,I^*)$, which is derived from the main model (\ref{model1}):  
	\begin{equation}
	\label{EE1}
   \Lambda- (\mu+c_S \phi(I))S =(\mu +\nu + c_I \psi(I)) I.
	\end{equation}
Rearranging the equality above, we obtain 
\begin{equation}
\label{equa1}
S=\frac{\Lambda-(\mu +\nu + c_I \psi(I)) I}{(\mu+c_S \phi(I))}.
\end{equation}
Also notice that an equilibrium condition of the system is
\begin{equation}
\label{bS}
 \beta S=(\mu +\nu + c_I \psi(I))(S+I),
\end{equation}
where $S+I = N$.\\
We denote the function of $I$ representing $S$ by $f(I)$. Then 
\begin{align}
\label{equi3}
\frac{\beta f(I)}{I + f(I)}=\mu +\nu + c_I \psi(I).
\end{align}     
After substituting (\ref{equa1}) into (\ref{equi3}), we get the following equality:
\begin{align}
\label{equalast}
  \beta = \frac{(\mu+c_S\phi(I))(\mu +\nu + c_I \psi(I))I}{\Lambda-(\mu +\nu + c_I \psi(I))I}+\mu +\nu + c_I \psi(I) 
\end{align}
Define $\beta^0$ as the critical value of $\beta$ for which $\mathcal
R_0=1$, i.e.
$$\beta^0 = \mu+\nu +c_I/B.$$
Let $F(I)$ be the right hand side of the equality (\ref{equalast}).  Equation (\ref{equalast}) defines the infected equilibrium $I^*$ implicitly as a function of $\beta$, with $\beta=F(I^*)$.  Notice that $\beta^0=F(0)$.  Hence, the implicit derivative $\frac{dI^*}{d\beta}$ evaluated at $I^*=0,\beta=\beta^0$ gives the direction of the bifurcation at $\mathcal R_0=1$.  If
$$\frac{dI^*}{d\beta}|_{(I^*=0,\beta=\beta^0)}>0,$$
then the bifurcation is forward, and, conversely, the reverse inequality implies backward bifurcation.  Taking this implicit derivative, we obtain:
\begin{align*}
1&=F'(0)\frac{dI^*}{d\beta}|_{(I^*=0,\beta=\beta^0)}
\end{align*}
Hence, the sign of $F'(0)$ determines the direction of bifurcation, where
\begin{align}
\label{fbb}
F'(0)&=c_I\psi'(0)+\frac{(\mu+\nu+c_I\psi(0))(\mu+c_s\phi(0))}{\Lambda}.
\end{align}
The result follows.
\end{proof}
\subsection{Example Culling Rates}
In the following sections, we consider three constitutive forms for the culling rates:
\begin{itemize}
\item {\underline {Mass culling rates}}: 
 In the case of mass culling, the per capita culling rate increases with the number of infected $I$.  Hence, we suppose that $\phi(I)$ and $\psi(I)$ are increasing functions, i.e. $\phi'(I),\psi'(I)\geq 0$ for $I\geq 0$.  Although our results are proved for general increasing per-capita culling rates, we remark that restricting to the case $\phi(0)=\psi(0)=0$ may be biologically reasonable.  The fact that culling is employed after the onset of an outbreak, along with the possible detection problems at low levels of infection, provides motivation for assuming that as $I\rightarrow 0$, the per capita culling rates decrease to zero. If the assumption $\psi(0)=0$ is made, then $\mathcal R_0$ is independent of the culling rates.

 \item {\underline {Modified culling rates}}: Per capita culling rates
$c_S\phi(I)$, $c_S\psi(I)$ are non-monotone functions of infected bird population. They increase with respect to infected bird population once the outbreak occurs, but when infected bird population gets sufficiently large, the per capita culling rates strictly decrease. Decreasing in the per-capita culling rates may occur since resources for carrying out that control measure are limited. As an example, we model with   
\begin{align*} 
 c_S\phi(I)=c_S\frac{I}{A+I^2},\quad c_I\psi(I)=c_I\frac{I}{B+I^2},
\end{align*} 
  where $A,B >0$.  Because of the reactionary nature of culling and detection problems mentioned in the above paragraph, the assumption that $\phi(0)=\psi(0)=0$ is reasonable, and $\mathcal R_0$ does not depend on culling.
     
\item  {\underline {Selective culling rates}}: For the case of selective culling, we assume that the ability of farmers to discriminate infected from susceptible birds is perfect and only infected birds are culled. We let $c_S\phi(I)= 0$ and $c_I\psi(I)=\displaystyle c_I\frac {1}{B+I}$, with $B>0$. Hence the per-capita culling rate $c_I\psi(I)=\displaystyle c_I\frac {1}{B+I}$ decreases as the number of infected birds increases. Saturation in total number of infected birds culled occurs as a result of limited culling effort. Moreover, notice that $\psi(0)>0$, and hence $\mathcal R_0$ depends on $c_I$. In the small family-run farms in which selective culling is utilized, there might be better detection of H5N1 when the number of infected is small. Therefore, we are interested in exploring the dynamics in the case where $\psi(0)>0$.  Selective culling may be employed in poultry because of socio-economic impact of mass culling.
\end{itemize}

\section{Mass Culling}
In this section we consider and analyze the model \eqref{model1} with increasing per-capita culling rates $\psi(I), \ \phi(I) :$ 
\begin{align}
\label{mass} 
\psi'(I)\geq 0, \ \phi'(I) \geq 0, \ \forall \ I \geq 0. 
\end{align}
Recall the reproduction number 
\begin{equation}
\mathcal R_0 = \frac{\beta }{\mu+\nu + c_I \psi(0)}
\end{equation}
and the disease free equilibrium $\mathcal E_0=(\Lambda/\mu,0)$. Now we want to show existence of a unique endemic equilibrium when $\mathcal R_0>1$ and in addition, we will show that when $\mathcal R_0<1,$ the system with mass culling (\ref{mass}) does not have an endemic equilibrium.
By the equation (\ref{equalast}), we have
$$\Lambda - (\mu + \nu + c_I  I)
I = \D \frac{ (\mu+c_S \phi(I))(\mu + \nu + c_I\psi( I))
  I}{\beta -(\mu + \nu + c_I \psi (I)) } .$$
Let the left and right hand side of this equation be $F(I)$ and $G(I)$, respectively. $F(I)$ is a decreasing function of $I$ with $F(0)=\Lambda>0$ and $\lim_{I \rightarrow \infty} F(I)= - \infty $.  Notice that $G(0)=0$, and if $\mathcal R_0>1,$ then either $G(I)$ has a vertical asymptote at a point $I_c \in (0, \infty)$: $\beta -(\mu + \nu + c_I \psi (I_c))=0$ or  $\beta -(\mu + \nu + c_I \psi (I))>0, \ \forall I>0.$ If $\beta -(\mu + \nu + c_I \psi (I))>0, \ \forall I>0,$ then by taking derivative of $G(I)$ with respect to $I$, one can easily see that $G'(I)>0$ for all $I$ since $\phi'(I),\psi'(I)\geq 0$. Then the equality above has a unique positive solution since $F(0)> G(0).$ Now suppose there exists $I >0: \ \beta -(\mu + \nu + c_I \psi (I))=0$. Let  $I_c$ be the minimum positive root such that $\beta -(\mu + \nu + c_I \psi (I_c))=0.$ For any endemic equilibrium $(S^*, I^*)$ of the system (\ref{model1}), $S^*$ takes positive value only if $ I^* \in \left[0, I_c\right)$, by the equilibrium condition (\ref{bS}). In this case, since $F(0)>G(0),$ it is enough to show $G(I)$ is increasing on the intervial $ \left[0, I_c\right)$ and $\lim_{I\rightarrow I_c^-} G(I)=+\infty$, which is easy to see. Therefore if $\mathcal R_0>1$, then there exists a unique endemic equilibrium. On the other hand, if $\mathcal R_0<1,$ then the same equilibrium condition (\ref{bS}) excludes the possibility of presence of any endemic equilibrium for mass culling rates since $\psi'(I)\geq 0$.  \\
To consider the stability of the endemic equilibrium, we look at the Jacobian
\begin{equation}
\label{J}
J=
\left(
\begin{array}{ll}
-\D\frac{\beta I}{N}  +\D\frac{\beta IS}{N^2} - \mu - c_S \phi(I)
                 & -\D\frac{\beta S}{N} +\D\frac{\beta IS}{N^2}
- c_SS \phi'(I)\\
\\
\D\frac{\beta I}{N}  -\D\frac{\beta IS}{N^2} & 
        \D\frac{\beta S}{N} -\D\frac{\beta IS}{N^2} -
        (\mu+\nu+c_I (\psi(I)+ \psi'(I)I))\\
\end{array}
\right)
\end{equation}
To see that the endemic equilibrium is locally asymptotically stable, we look
at the signs of the entries of $J$ when $J$ is evaluated at the endemic
equilibrium.
Simplifying entry $j_{11}$ we can show that $j_{11}<0$. Similarly entry
$j_{21}>0$ and entry $j_{12}<0$.
By an equilibrium condition for the system  \eqref{model1}, we obtain that
entry
$j_{22}<0$. Hence, $Tr\ J = j_{11}+j_{22} <0$. Furthermore, $Det\ J =
j_{11}j_{22}-j_{12}j_{21}>0$.
Therefore the endemic equilibrium is locally asymptotically stable whenever it
exists. We summarize these results in the following Theorem:

\begin{theorem} With mass culling (\ref{mass}), the system \eqref{model1} has a unique endemic equilibrium,
  if $\mathcal R_0>1$. However, when $\mathcal R_0<1, $ the system does not exhibit any endemic equilibrium. Moreover the endemic equilibrium $\mathcal
E^*$ is locally asymptotically stable whenever it exists.
\end{theorem}
Furthermore, the following result can be established by Theorem \ref{GAS}:
\begin{theorem} If $\mathcal R_0<1$, then the disease-free equilibrium $\mathcal E_0=\left(\Lambda/\mu,0 \right)$ is
  globally stable.
\end{theorem}
Since the endemic equilibrium is unique and locally stable, it is reasonable  to
expect that it is globally stable. Indeed, we have the following result
\begin{lemma} With mass culling (\ref{mass}), the system \eqref{model1} has no periodic orbits in the first quadrant.
\end{lemma}

\begin{proof}
We apply Dulac's criterium with Dulac function $ D =1/SI$. Let $X$ be the
open first quadrant and $f_1$, $ f_2$ be defined as follows:
\begin{align*}
\begin{cases}
f_1 & \hspace{-4mm}= \Lambda -\displaystyle\frac{\beta I S}{N} - \mu
S-c_S \phi(I)S  \vspace{1.5mm},\\
f_2 & \hspace{-4mm}=\displaystyle\frac{\beta I S}{N} -  (\mu +
\nu)I - c_I \psi(I) I\vspace{1.5mm}\\
\end{cases}
\end{align*}
Then 
$$\D\frac{\partial D f_1}{\partial S} + \frac{\partial D f_2}{\partial I}
= -\left(\frac{\Lambda}{S^2 I} +\frac{c_I \psi'(I)}{S}\right) <0, \ \forall \ I \geq 0.$$
Thus, Dulac's criterium implies that there are no periodic orbits in the first quadrant.
\end{proof}

\begin{theorem} With mass culling (\ref{mass}) the system \eqref{model1} has  a
 unique endemic equilibrium which is globally stable if $\mathcal R_0>1$.
\end{theorem}

\begin{proof} To obtain this result one needs to apply the Poincar\'{e}-Bendixson Theorem.  First, recall that all solutions of  system \eqref{model1}
  are bounded. Indeed we have
$$\limsup_t (S+I) \le \limsup_t N \le \D \frac{\Lambda}{\mu}.$$
Since when $\mathcal R_0>1$  the disease-free equilibrium is a unstable saddle
with
stable manifold along the $S$ axis, then solutions that start from
$I(0)=0$ will stay on this stable manifold and converge to the disease-free
equilibrium.
If, however, $I(0)>0$, then the solution is repelled by the disease-free
equilibrium
and the omega limit set must contain another equilibrium (since there are no
periodic orbits). The only option is the endemic equilibrium. Once the
solution gets close to the endemic equilibrium which is locally stable, it
will attract the solution. Thus every solution for which $I(0)>0$ converges
to the endemic equilibrium. Hence, the endemic equilibrium is globally stable.
\end{proof}

\begin{figure}
\begin{center}
\includegraphics[width=9cm,height=5cm]{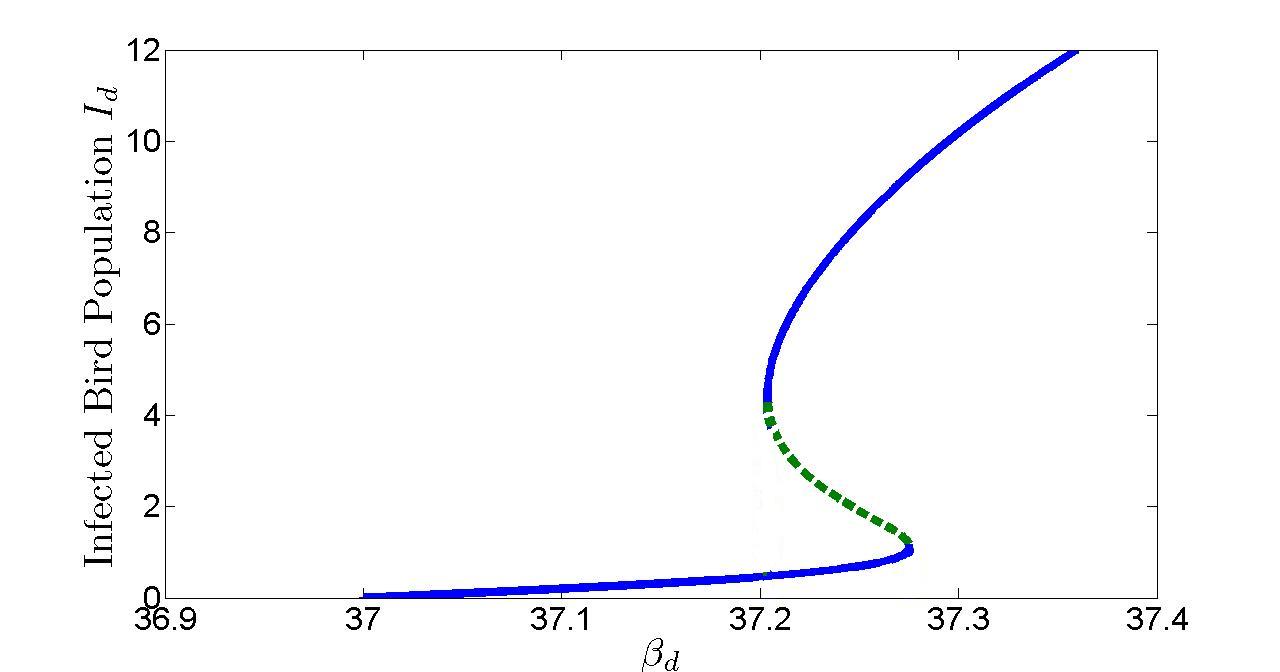}
\end{center}
\caption{\label{fig:hysteresis}
Hysteresis caused by the per capita culling rates $c_S\phi(I)=c_S\frac{I}{A+I^2}$ and $c_I\psi(I)=c_I\frac{I}{B+I^2}$ in the model
\eqref{model1}. Parameter values are: $A = 1$, $B = 1$,
$\nu = 0.1*365$,
$\mu = 0.5$,
$c_S = 0.5$,
$c_I = 0.5$,
$\Lambda = 1200$.}
\end{figure} 

\section{Modified culling}
Mass culling has proved effective for isolated outbreaks, but is a less successful control measure in more widespread outbreaks. If the number of infected becomes large enough, disease control authorities can rapidly become overwhelmed through lack of resources. In the case of widespread infection, it has been suggested that disease control should be shifted from a traditional culling approach to a modified strategy which entails culling of only infected and high-risk in-contact poultry along with complimentary measures such as vaccination. This shift in strategy may mitigate the economic costs of mass culling. Modified culling, therefore, comprises of sustained culling in isolated outbreaks and a decrease in total culling effort during a widespread outbreak \cite{FAO0, FAO00}. This motivates us to explore the dynamical consequences of considering per capita culling rates $c_S\phi(I)$ and $c_I\psi(I)$ which are proportional to $I$ when $I$ is small, but decreasing when $I$ is large. We assume that the per capita culling rates $c_S\phi(I)$, $c_I\psi(I)$ in system (\ref{model1}) satisfy the following conditions:
\begin{align}
\label{COND2}
 & i) \ \psi(0)=0, \phi(0)=0. \nonumber \\
 & ii) \ \psi'(0)>0, \phi'(0)>0. \\
 & iii) \text{ Once $I$ is sufficiently large, }  \psi(I), \phi(I) \text{ strictly decrease. } \nonumber 
\end{align}
 Note that our model system \eqref{model1} considers the scenario where no other control measures beside culling are implemented. 
  
 From the equation (\ref{RN}), the reproduction number in the case of modified culling is 
       \begin{equation}
\mathcal R_0 = \frac{\beta }{\mu+\nu}.
\end{equation}

From Theorem \ref{GAS}, we have the following result about the global stability of $\mathcal E_0$ in the case of modified culling:

\begin{theorem} If $\mathcal R_0< 1$, then the disease free equilibrium $\mathcal E_0$ is globally asymptotically stable. 
  \end{theorem}
  
 Next, we obtain the following result:
  
 \begin{figure}
\begin{center}
\includegraphics[width=12cm,height=9cm]{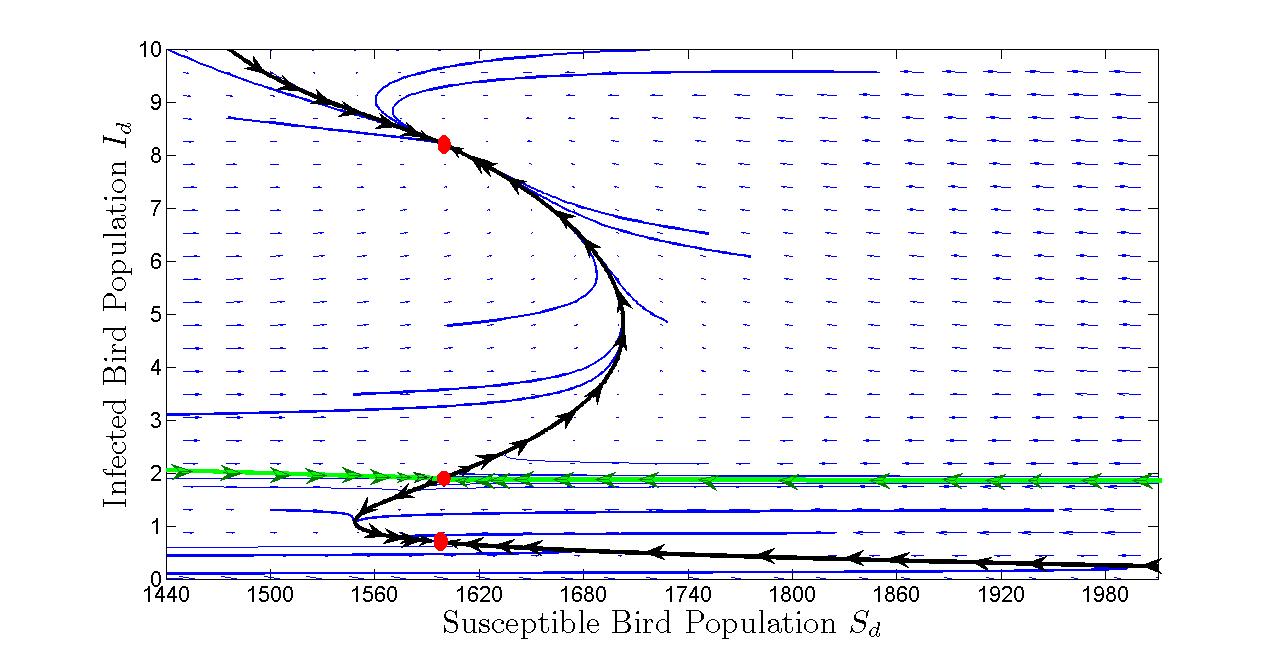}
\end{center}
\caption{\label{fig:phaseplan} Phase portrait of the model
\eqref{model1} with the per capita culling rates $c_S\phi(I)=c_S\frac{I}{A+I^2}$ and $c_I\psi(I)=c_I\frac{I}{B+I^2}$. Parameter values are: $A = 1$, $B = 1$,
$\nu = 0.1*365$,
$\mu = 0.5$,
$c_S = 0.5$,
$c_I = 0.5$,
$\beta = 37.25$,
$\Lambda = 1200$.}
\end{figure}

 \begin{theorem} For the system with modified culling (\ref{COND2}), there is always a forward bifurcation at $\mathcal R_0=1$.
\end{theorem} 

\begin{proof} For the system with modifying culling (\ref{COND2}), we have $\psi'(0)>0.$ Then the proof follows the theorem ($\ref{fb}$). 
\end{proof}
  
As an example for per capita culling rates satisfying the conditions in (\ref{COND2}), we consider the functions  
\begin{equation}
\label{modifiedcullrate} 
 c_S\phi(I)=c_S\frac{I}{A+I^2},\quad c_I\psi(I)=c_I\frac{I}{B+I^2},
\end{equation} 
where $A,B >0$.\\  
To analyze the resulting system with the per-capita culling rates above in a more convenient way, we take $A=B$. Then after simplification, from the equality $(\ref{EE1})$, we obtain a fifth degree polynomial in $I$. Then we study the number of positive roots of the fifth degree polynomial: 
$$F(I)= xI^5 + y I^4 + zI^3 + lI^2 + mI + n,$$
where $$x=\beta(-\beta+\nu)(\mu+\nu), \ y=(n/ A^2)+f, \ z= 2Ax+g,\ l= (2n/A)+Af,$$
			$$m=A( Ax-\beta \Lambda c_I), \ n=A^2 \Lambda \beta (\mu+\nu) (\mathcal R_0- 1)$$
and \\
      $$g=-\beta \Lambda c_I-\beta c_I c_S+\beta c_I^2, \ f=-\beta^2c_I-\beta c_S(\mu+\nu)+\beta c_I(\mu+\nu)+c_I\nu \beta $$ \\
Next, we obtain the following result:
    
\begin{proposition} \label{prop}
  If $\mathcal R_0>1$, then there is either a unique or three or five positive endemic equilibria, if all equilibria are simple roots.
\end{proposition} 
\begin{figure}
\begin{center}
\includegraphics[width=10cm,height=5cm]{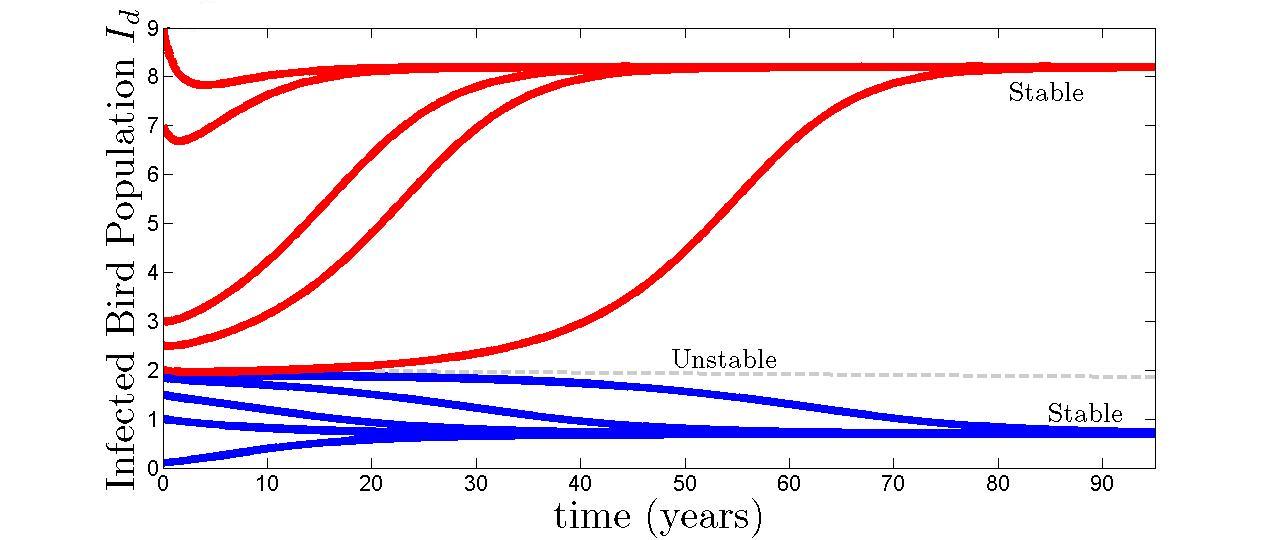}
\end{center}
\caption{\label{fig:dynamics}
Simulation of infected bird population versus time for various initial conditions and the same parameter values as in Fig.\ref{fig:phaseplan}.}
\end{figure}

\begin{proof} Suppose  $\mathcal R_0>1.$ Then the leading coefficient $x$ is negative since this implies $-\beta+\nu<0$. Hence $$\lim_{I \rightarrow \infty}F(I)=-\infty.$$ Also note that $F(0)=n$ and $n>0$ when $\mathcal R_0>1.$ $F(I)$ is a continuous function of $I$ and by fundamental theorem of algebra, we know that this polynomial can have at most five real roots. Through a geometric argument, now it is easy to see that there is either a unique or three or five positive endemic equilibria, if all equilibria are simple roots.
\end{proof}

    The bifurcation diagram in Fig.\ref{fig:hysteresis} shows the existence of a forward hysteresis bifurcation under some certain parameter values. Estimates of the parameters $\Lambda$, $\mu$, and $\nu$ are derived in \cite{Iwami1, Martcheva5}. For the transmission and culling rates, there is insufficient information, and these parameter values will be varied in the simulations. Note that the units for $N$ is $birds\ \times 10^7$. In simulations, we see that when $\mathcal R_0<1$, the disease dies out, but when $\mathcal R_0>1$, such a bifurcation may display certain catastrophic behaviors: a solution with initially a small number of infected birds may converge to the equilibrium level with a large number of infected birds. We discuss the epidemiological implications of forward hysteresis in Section $6$.

To find the stability of equilibria, we evaluate the Jacobian of the system
\eqref{model1} with non-monotone per capita culling rates. The general form of Jacobian of the system is given in (\ref{J}).
  We obtain the following result:
    
  \begin{theorem}\label{t7} Suppose $\mathcal R_0>1$ and there are three or five endemic equilibria. Suppose further the endemic equilibria are ordered with respect to the number of infected. Then the even numbered endemic equilibria are always unstable. 
  
\end{theorem}

\begin{proof}

To investigate the stability of the  endemic equilibria, first we determine the sign of the determinant of $J$
depending on the equilibrium at which it is evaluated. From the equation \eqref{equi3} for
the equilibria, we have 
\begin{equation}\label{Ineq}
\D\frac{\beta f'(I)}{f(I)+I} -
\D\frac{\beta f(I)(f'(I)+1)}{(f(I)+I)^2}< - \D\frac{c_I (I^2-B)}{(B+I^2)^2} 
\end{equation}  
The inequality ``$<$'' holds in the case of the unique endemic equilibrium when $\mathcal R_0>1$ and in the case of the odd numbered endemic equilibria if there are three or five endemic equilibriums when $\mathcal R_0>1$. For the even numbered endemic equilibria, the inequality is
exactly the opposite, that is, we  have ``$>$''. We rewrite the above
inequality in the form
\begin{equation}\label{Ineq12}
\left[\D\frac{\beta}{N} -
  \D\frac{\beta S}{N^2}\right] f'(I) < \D\frac{\beta S}{N^2} - \D\frac{c_I(I^2-B)}{(B+I^2)^2} 
\end{equation}
where $N =f(I)+I$ and $S=f(I)$. In the expression above, $f'(I)$
is given by
\begin{equation}
f'(I) = \D\frac{-\left(\mu+\nu+ c_I\D\frac{2I
    B}{(B+I^2)^2}\right)}{\mu+c_S\D\frac{I}{A+I^2}}
+ \frac{Sc_S \D\frac{(I^2-A)}{(A+I^2)^2}}{\mu+c_S\D\frac{I}{A+I^2}}
\end{equation}
To simplify the $Det\ J$ we introduce two notations:
$$E_1 =  \D\frac{\beta S}{N} -\D\frac{\beta IS}{N^2},\ E_2 = \D\frac{\beta I}{N}  -\D\frac{\beta IS}{N^2}.$$
With this notation the determinant becomes:
$$Det\ J = \left[E_1-\left(\mu+\nu + c_I\D\frac{2IB}{(B+I^2)^2}\right)\right]
\left[-E_2 - \mu-c_S\D\frac{I}{A+I^2}\right] - E_2\left[-E_1 + c_SS\D\frac{I^2-A}{(A+I^2)^2}\right]$$
Simplifying the determinant, we have
\begin{align}
Det\ J &= - E_1 \left(\mu+c_S\D\frac{I}{A+I^2}\right) + E_2\left(\mu+\nu + c_I\D\frac{2IB}{(B+I^2)^2}\right)\nonumber\\
\qquad\qquad & + \left(\mu+\nu + c_I\D\frac{2I
    B}{(B+I^2)^2}\right)\left(\mu+c_S\D\frac{I}{A+I^2}\right) -E_2
c_S S\D\frac{I^2-A}{(A+I^2)^2}.\nonumber
\end{align}
Factoring out $\left(\mu+c_S\D\frac{I}{A+I^2}\right)$ we may recognize that
$E_2$ is multiplied by the derivative $f'(I)$. Simplifying further, we have
$$Det\ J = \left(\mu+c_S\D\frac{I}{A+I^2}\right)\left[-E_1 -E_2 f'(I) + \mu+\nu + c_I\D\frac{2IB}{(B+I^2)^2}\right].$$
Inequality \eqref{Ineq12} now implies that 
$$E_2f'(I) <\D\frac{\beta S I}{N^2}-c_I I\D\frac{I^2-B}{(B+I^2)^2}.$$
Replacing $E_2f'(I)$ with the right side of the above inequality in $Det\ J$,
we obtain
$$Det\ J > \left(\mu+c_S\D\frac{I}{A+I^2}\right) \left[-\D\frac{\beta
    S}{N} + \mu +\nu+c_I \frac{I}{B+I^2}\right] = 0.$$
The last equality follows from the equation for the equilibria.
Hence, $Det\ J>0$ for the unique equilibrium when $\mathcal R_0>1$ and for the
odd numbered endemic equilibriums when $\mathcal R_0>1$ and three or five endemic equilibriums exist. In addition $Det\ J<0$ for the even numbered endemic equilibriums when three or five equilibria
exist. Hence, the even numbered endemic equilibriums are always unstable saddle points. 

\end{proof}

  In Fig.\ref{fig:phaseplan}, a phase portrait of the main model \eqref{model1} with the per capita culling rates $\phi(I)=\frac{I}{A+I^2}$ and $\psi(I)=\frac{I}{B+I^2}$ is numerically generated using $Matlab$. Parameter values are chosen as in Fig.\ref{fig:hysteresis} and $\beta = 37.25$. The three equilibria are marked as red dots. The heteroclinic orbits from the middle equilibrium to the upper one and the lower one are displayed as black curves, so their basin of attraction is clearly identifiable. They are the unstable manifolds of the saddle point (the middle equilibrium). Also the stable manifolds of the saddle point are displayed as green curves. When $\mathcal R_0>1$ and there are three endemic equilibria, the region of attraction to which the initial condition belongs plays a crucial role. If the initial point lies in the basin of attraction of the lower equilibrium, the infection persists at a low equilibrium level. However, if the initial point lies in the basin of attraction of the upper equilibrium, it persists at the equilibrium level with the largest number of infected birds.

\begin{figure}
\begin{center}
\includegraphics[width=12cm,height=5cm]{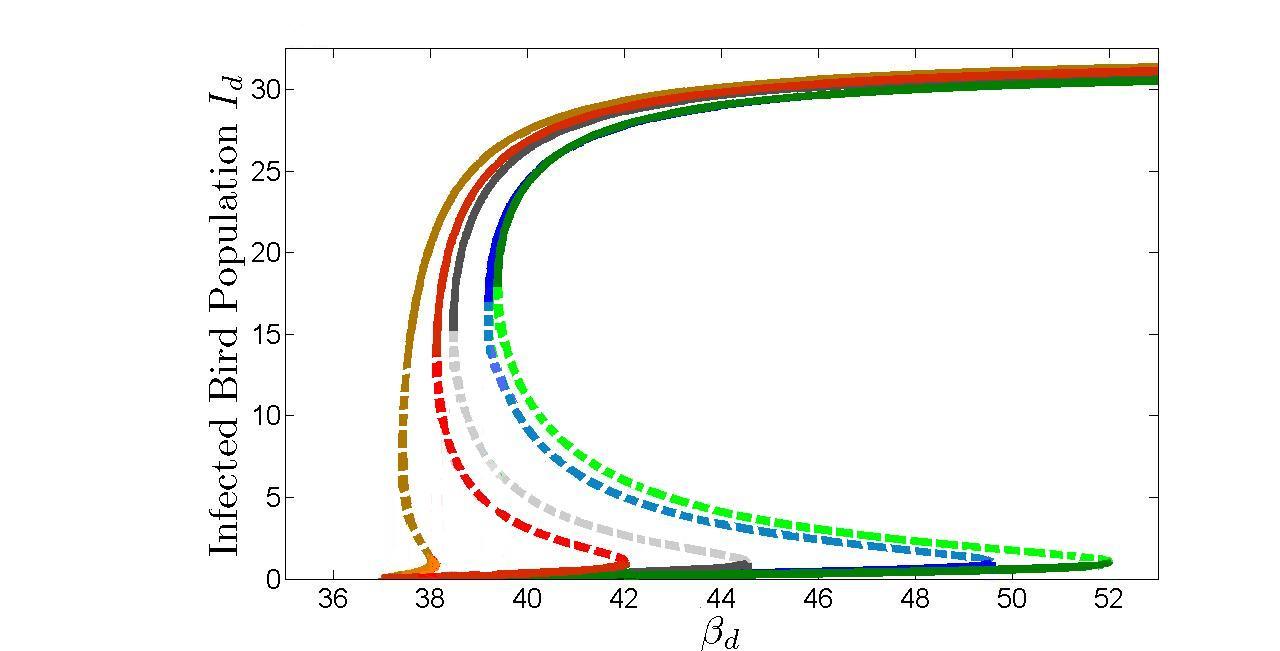}
\end{center}
\caption{\label{fig:region} The bifurcation diagram with respect to the parameter $\beta$ for the distinct per capita culling coefficients $c_I = 2 \text{(orange)}, 10 \text{ (red)}, 15 \text{ (grey)}, 25 \text{ (blue)}, 30 \text{ (green)}$. Observe that increasing per capita culling coefficient increases the region of hysteresis. The parameter values are the same with Fig.\ref{fig:phaseplan}}
\end{figure}
     
    For given parameter values in Fig.\ref{fig:phaseplan}, bistability occurs, where both the lower and upper equilibrium are attractors. This can be seen in simulations in Fig.\ref{fig:dynamics}. For solutions close to the unstable equilibrium, there is a sizable lag in time before they are repelled from the unstable equilibrium and converge to a stable equilibrium. Furthermore, Fig.\ref{fig:hysteresis} shows that multiple equilibria occur with $\beta$ as a bifurcation parameter. We are interested how the culling coefficient $c_I$ of infected birds affects the strength of the hysteresis. Fig.\ref{fig:region} presents the bifurcation diagram with respect to the parameter $\beta$ for various per capita culling coefficients. We see that an increase in the per-capita culling coefficient $c_I$ increases the region in which multiple equilibria occurs and the basin of the attraction of lower equilibrium.

\section{Selective culling}

  Culling exclusively the known infected and high risk in-contact birds has been utilized as a control measure in some countries such as Vietnam, Indonesia, Thailand and Cambodia. One of the most important reasons behind the use of this type of culling is to mitigate the heavy economic losses associated with mass culling. This strategy is used especially in the countries dominated by backyard poultry farms such as Vietnam \cite{ FAO0, FAO00, Modcull, WHO0, Sridhar}. Usually these farms are owned by families, with low incomes. In some cases, even though the disease control authorities do only approve mass culling, it is well-known that farmers are avoiding to go through mass culling. Economic concerns prompt farmers to almost exclusively cull infected birds. In order to model this scenario, we consider an extreme case in which the target group for culling only includes infected birds. Furthermore, we assume resource limitation causes a saturation in the rate of culling of infected poultry. Therefore, we explore the dynamics of the system \eqref{model1} with per capita culling rates $c_S\phi(I)$ and $c_I\psi(I)$ where $c_S=0$ and $c_I\psi(I)$ is a decreasing function of $I$. We analyze the system \eqref{model1} with selective culling rates 
 \begin{equation}
 \label{selectivecullrate}
  c_S\phi(I)S =0,\quad c_I\psi(I)I=c_I\D\frac{I}{B+I}, 
 \end{equation} 
 with $B>0$. 
Applying (\ref{RN}), the reproduction number in the case of selective culling is derived as
\begin{equation}
\mathcal R_0 = \frac{\beta }{\mu+\nu +c_I/B}.
\end{equation}
Note that in this case, the reproduction number depends on the culling rate $c_I$. \\
\begin{figure}
\begin{center}
\includegraphics[width=9cm,height=5cm]{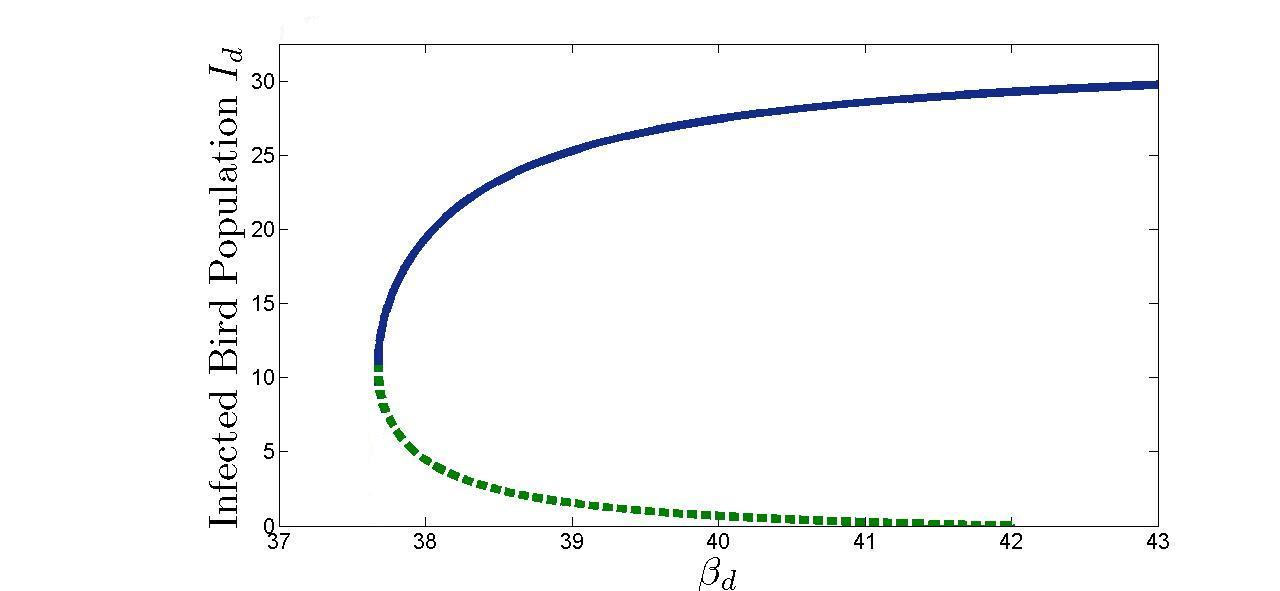}
\end{center}
\caption{\label{fig:backward}
Backward bifurcation caused by selective culling in model
\eqref{model1}. Parameter values are: $B = 1$,
$\nu = 0.1*365$,
$\mu = 0.5$,
$c_I = 5$,
$\Lambda = 1200$.}
\end{figure}

\begin{theorem} If  $\mathcal R_0>1$, there is a unique endemic equilibrium  $\mathcal
  E^*=(S^*_d,I^*_d)$.
If $\mathcal R_0<1$, then there may be zero endemic equilibria or backward bifurcation may occur in which case there can be two equilibria
$\mathcal E_1^* =(S^*_{d1},I^*_{d1})$ and $\mathcal E_2^* =(S^*_{d2},I^*_{d2})$.
\end{theorem}

\begin{proof} First we consider the case $\mathcal R_0>1$. We rewrite equation
  \eqref{equi3}
in the form of an equality of two polynomials $p_1(I) = p_2(I)$ where
\begin{align}
p_1(I)& = [\Lambda(B+I)
  -((\mu+\nu)(B+I)+c_I)I][(\beta-\mu-\nu)(B+I) -c_I]\nonumber\\
p_2(I)& = \mu(B + I)[(\mu+\nu)(B+I)+c_I]I\nonumber
\end{align}

$p_1(I)$ is a polynomial of degree three satisfying:
\begin{itemize}
\item $p_1(0) >0$.
\item $p_1(I)$ has one positive and two negative roots.
\item $\lim_{I\to\infty} p_1(I) = -\infty$.
\end{itemize}
At the same time $p_2(I)$ for $I\ge 0$ is a strictly increasing
function. Furthermore
$p_2(0)=0$. Under these conditions, it is not hard to see that the two
polynomials always have exactly one intersection with $I>0$.

In the case $\mathcal R_0<1$,  $p_1(I)$ is a polynomial of degree three satisfying:
\begin{itemize}
\item $p_1(0) <0$.
\item $p_1(I)$ has at most two positive roots.
\end{itemize}
At the same time $p_2(I)$ for $I\ge 0$ has the same properties. Hence, if
there are intersections of the two polynomials for $I>0$, these
intersections are at most two. 
\begin{figure}
\begin{center}
\includegraphics[width=12cm,height=9cm]{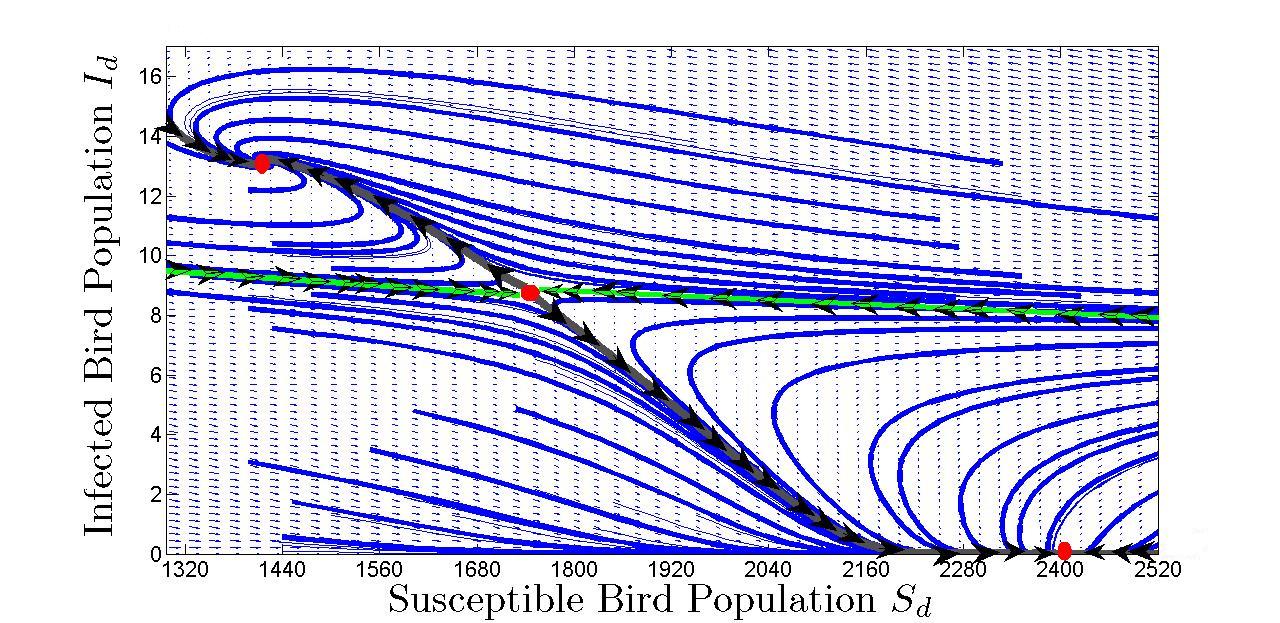}
\end{center}
\caption{\label{fig:phasebackward}
Phase Portrait of the model
\eqref{model1} with selective culling rate $c_I\psi(I)I=c_I\frac{I}{B+I}$. Parameter values are: $B = 1$,
$\nu = 0.1*365$,
$\mu = 0.5$,
$c_I = 5$,
$\beta = 37.7$
$\Lambda = 1200$.}
\end{figure}
Two  subthreshold endemic equilibria are present in the case when backward
bifurcation occurs. By the equation (\ref{fbb}), we have
\begin{equation}
F'(0) = \D\frac{-c_I \Lambda +  \mu[(\mu + \nu)B^2 + c_IB]}{B^2 \Lambda}
\end{equation}
Hence, by the theorem (\ref{fb}), backward bifurcation occurs if and only if $F'(0)<0$, that is if and
only if
\begin{equation}\label{BBC}
 -c_I \Lambda +  \mu[(\mu + \nu)B^2 + c_IB]<0.
\end{equation}
\end{proof}

\begin{figure}
\begin{center}
\includegraphics[width=9cm,height=5cm]{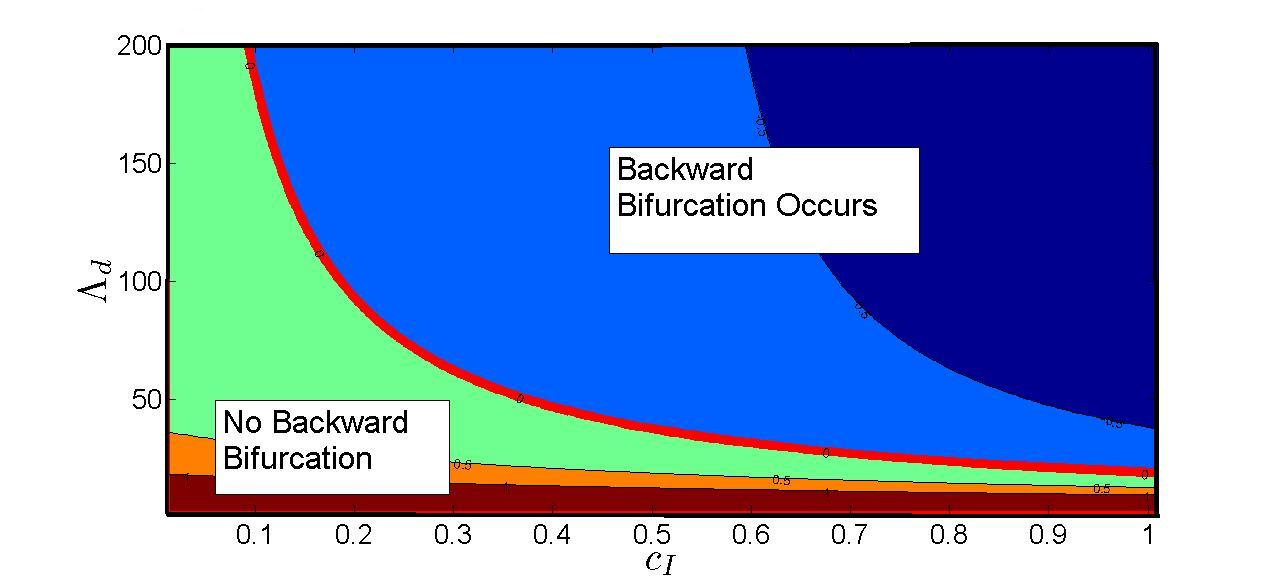}
\end{center}
\caption{\label{fig:Area}
Area in the $\Lambda$, $c_I$ space where backward bifurcation occurs in model
\eqref{model1}. The red line separates the two regions where backward
bifurcation occurs and where backward bifurcation does not occur. Blue  and dark blue color is
associated with backward bifurcation.
Parameter values are: $B = 1$,
$\nu = 0.1*365$,
$\mu = 0.5$,
$\beta = 37.7$.} 
\end{figure}
We illustrate the occurrence of backward bifurcation in the Fig.
\ref{fig:backward}. 

The area of the parameter space where condition \eqref{BBC} holds
is illustrated in Fig.\ref{fig:Area}.
 Fig.\ref{fig:Area} suggests that backward bifurcation always occurs when
 infected birds culling rate $c_I$ and susceptible birds repopulation rate
 $\Lambda$ are large enough. On the other hand, when the culling rate $c_I$ is small
 backward bifurcation does not occur. 

In what follows we investigate the stability of the equilibria.
To consider the stability of equilibria we look at the Jacobian of system
\eqref{model1} with selective culling.
Next, we derive the stability of equilibria:

\begin{theorem} \label{theo1} If $\mathcal
  R_0<1$ and there are two endemic equilibria, then $Det\ J < 0$ for the lower one (hence it is unstable) and $Det\ J > 0$ for the upper one. Moreover if $\mathcal R_0 > 1$, then $Det\ J > 0$ for the unique endemic equilibrium.
\end{theorem}
\begin{proof}
We use similar arguments with the proof of Theorem \ref{t7}. 
\end{proof}

 In fact, in the case $\mathcal R_0<1$, if $Tr\ J |_{\epsilon^*=(S^*,I^*)}<0$, the upper equilibrium $\epsilon^*=(S^*,I^*)$ is locally stable. However, numerical simulations suggest that for certain parameters, the upper equilibrium may lose stability via Hopf bifurcation, with an unstable bifurcating periodic solution. In this case, solutions outside the periodic orbit converge to the disease-free equilibrium. Furthermore, in the case backward bifurcation occurs and $\mathcal R_0>1$, numerically we found that $Tr\ J<0$ holds for all parameter values, but we could not prove that analytically.

\begin{theorem} \label{theo2} Assume the bifurcation is forward; that is 
$$
-c_I \Lambda +  \mu[(\mu + \nu)B^2 + c_IB]>0. 
$$
Then $Tr\ J |_{\epsilon^*=(S^*,I^*)}<0$ if $\mathcal R_0>1$.
\end{theorem}

\begin{proof}
By the equilibrium condition (\ref{bS}), we obtain
\begin{align*}
\begin{array}{l} 
Tr J  =   -\D\frac{\beta I}{N} - \mu + c_I\D\frac{I}{(B+I)^2}. \\ \\
\end{array}
\end{align*} 
If $\mathcal R_0 > 1$, then $ \beta = \mu + \nu + \frac{c_I}{B} + \epsilon $, for some $\epsilon > 0$. Therefore
\begin{align*}
\begin{array}{l} 
Tr J  \leq -\D\frac{I}{N} \left[\beta - \D\frac{c_I(\D\frac{\Lambda}{\mu})}{B^2} \right] - \mu \\  \\
\qquad\quad = -\D\frac{I}{N} \left[\mu + \nu + \frac{c_I}{B} + \epsilon - \D\frac{c_I\Lambda}{\mu B^2} \right] - \mu \\ \\
\qquad\quad = -\D\frac{I}{N}\epsilon - \D\frac{I}{N}\left[\mu \left[ (\mu + \nu )B^2 + c_I B \right] - c_I \Lambda \right].\frac{1}{\mu \beta^2}
\qquad\quad < 0.
\end{array}
\end{align*} 

\end{proof}

\begin{corollary} \label{cor1} If the bifurcation is forward, then the unique endemic equilibrium $\epsilon^*=(S^*,I^*)$ is locally asymptotically stable when $\mathcal R_0>1$.
\end{corollary}
\begin{proof} By using similiar argument in the proof of Theorem \ref{t7}, we show that $ Det\ J |_{\epsilon^*=(S^*,I^*)}>0$ whenever $\mathcal R_0>1$. Therefore Theorem \ref{theo2} implies the result in the Corollary \ref{cor1}.\\
\end{proof}

\begin{figure}[t]
\begin{center}
\includegraphics[width=13cm,height=5cm]{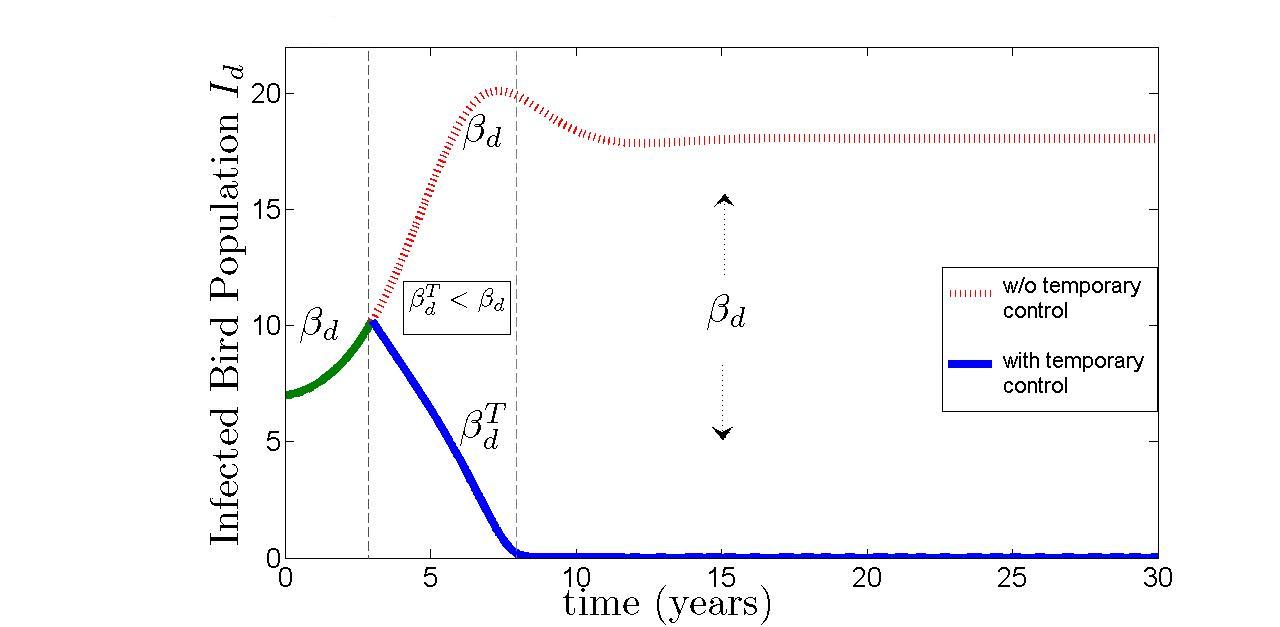}
\end{center}
\caption{
 Infected bird population versus time (in years) in which temporary control measures are employed along with selective culling (with the initial condition $(1000,7)$ and the same parameter values in Fig. \ref{fig:backward} and $\beta$ defined as a piecewise function \eqref{beta_d}). Also the solution is shown to converge to upper equilibrium for the case of no temporary control measures, i.e. with constant $\beta$. 
}
\label{fig:TMP}
\end{figure} 

\section{Epidemiological Implications of Bistable Dynamics}

 In the case of selective culling modeled by the per-capita culling rate (\ref{selectivecullrate}), the model undergoes a backward bifurcation as the transmission rate $\beta$ increases: the persistence of the disease can critically depend on the initial condition. For the parameter values utilized in Fig.\ref{fig:phasebackward}, numerical simulations show that the phase plane can be separated into two basins of attraction, one for the disease free equilibrium and the other for the upper equilibrium. Small changes in reproduction number can produce large changes in equilibrium behaviour: for initial conditions with an arbitrarily small number of infected birds, a sudden explosion occurs as $\mathcal R_0$ approaches 1. Given the bistability in the disease dynamics when $\mathcal R_0 < 1$, one can employ temporary control measures along with the culling in order to "push the solution" into the basin of attraction of the disease free equilibrium. Temporary control measures might include enhanced biosecurity, isolation of poultries from wild birds and movement ban of all poultry and hatching eggs. These measures have the effect of reducing the transmission rate $\beta$ for a given period of time. To model them, we can define
\begin{equation}
\label{beta_d}
{\beta(t) :=} \quad
\begin{cases}
 \beta,&\hspace{-2mm}  \text{if} \ \ \  \ 0 \leq t < t_1  \ \  \text{and} \  \ t_2 < t \vspace{1.5mm},\\
 \beta^T,&\hspace{-2mm}  \text{if} \ \ \  \ t_1 \leq t \leq t_2 \vspace{1.5mm},\\
\end{cases}
\end{equation}
with $\beta^T<\beta$.  \\
\\
If the duration or strength of the control measure is large enough, then the disease can be eradicated, as shown in Fig.\ref{fig:TMP}.

 \begin{figure}[t]  

 \subfigure[]{ 
 \includegraphics[width=0.99\textwidth,height=.4\textwidth]{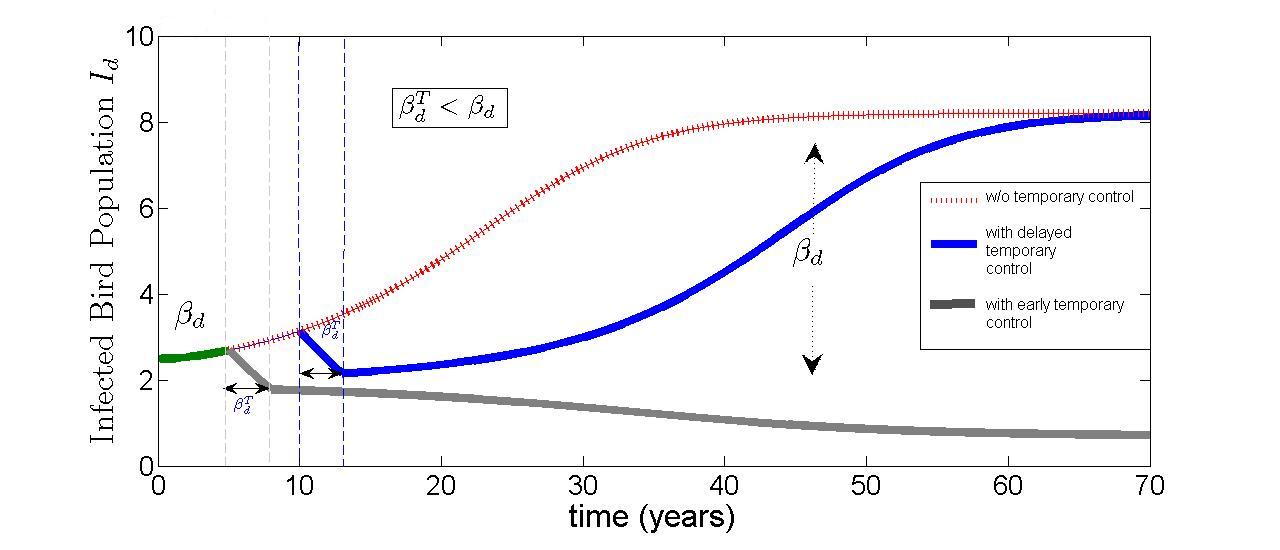}
 }
 \subfigure[]{ 
 \includegraphics[width=0.99\textwidth,height=.4\textwidth]{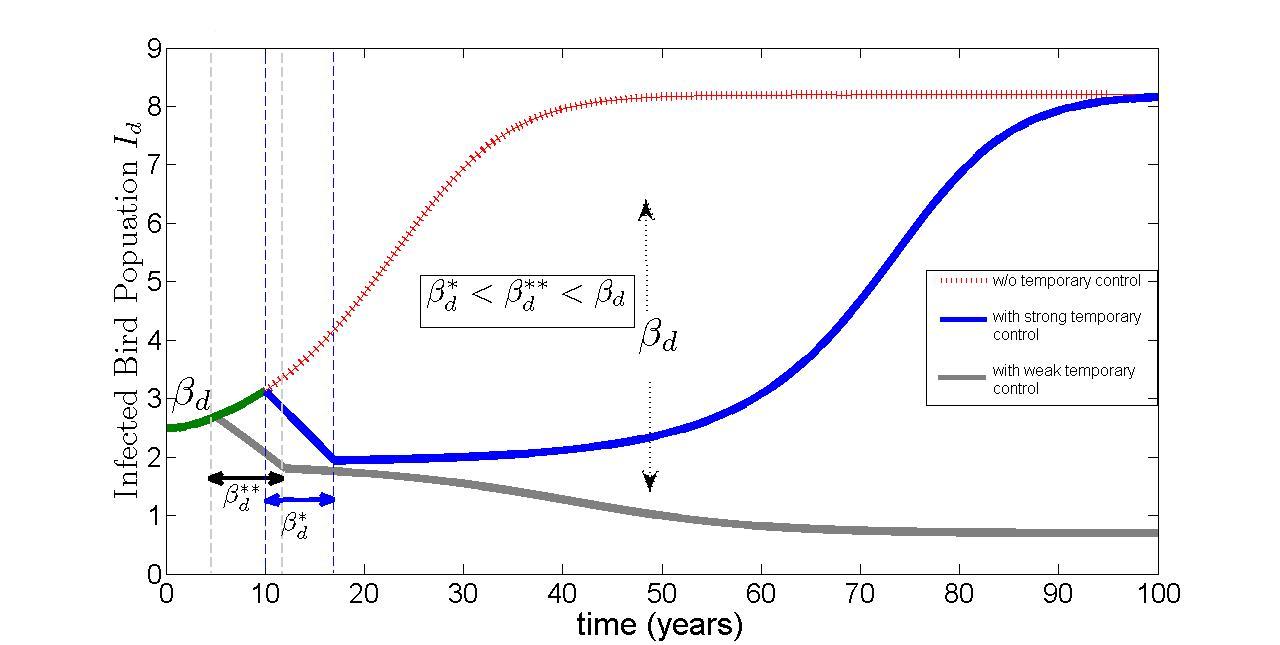}
 }
\caption[Optional caption for list of figures]{
Infected bird population versus time (in years) without temporary control measures or with distinct temporary control measures, along with modified culling. The initial condition is chosen as $(1000,2.5)$ and the same parameter values in Fig.\ref{fig:phaseplan} are used. Also, for temporary control measures, $\beta$ is defined as a piecewise function \eqref{beta_d}. Part (a) shows the solution of the delayed $\text{ (blue)}$ and the early $\text{ (grey)}$ temporary control measures, with the same duration and strength. In addition, part (b) shows the resulting solution of the delayed $\text{ (blue)}$ and the early $\text{ (grey)}$ temporary control measures, with the same duration but distinct strength.}
\label{fig:temporary}
 \end{figure}

  The system \eqref{model1} with modified culling rates (\ref{modifiedcullrate}) causes even a more complicated bifurcation: hysteresis. Even though backward bifurcation has been studied for a long time, hysteresis is less often detected in epidemiological models. The global dynamics of the disease for the parameter values in Fig.\ref{fig:phaseplan} shows that the lower and the upper equilibrium are co-existing attractors whose basin of attraction partition the feasible region. The equilibrium level of the disease critically depends on the location of the initial condition in the phase plane. When $\mathcal R_0>1$ and multiple attractive equilibria are present, introducing the temporary control measures may play a crucial role in keeping the number of infected at a low level and avoiding a sudden jump in the number of infected birds. 
  
  When no temporary control measures are applied, we see that a solution with an initial point in the basin of attraction of the upper equilibrium can experience a sudden jump to the upper equilibrium. However, if the temporary control measures decrease the transmission rate $\beta$ enough and are utilized for long enough time period, the number of infected birds decreases and the solution converges to the lower equilibrium which has the least number of infected birds. Simulations in cases with and without temporary control measures can be seen in Fig.\ref{fig:temporary}. Furthermore, an early application of temporary control measures can shorten the time period in which the temporary control must be applied in order to manage the disease. As can be seen in Fig.\ref{fig:temporary}(a), even though early temporary control and delayed temporary control have the same duration and strength, the solution with early temporary control converges to the lower equilibrium, but the solution of the delayed temporary control converges to the upper equilibrium. In fact in a given time period, early temporary control with weaker strength, i.e. less effective on reducing the transmission rate $\beta$, can be more efficient than delayed temporary control with larger strength. As can be observed in Fig.\ref{fig:temporary}(b), early temporary control measures with weaker strength reduces the number of infected birds so that the disease persists in a low level equilibrium. However, the solution of the delayed temporary control with larger strength converges to the upper equilibrium.

\section{Discussion}

  The emerging threat of a human pandemic caused by the H5N1 avian influenza virus strain magnifies the need for controlling the incidence of H5N1 infection in domestic bird populations. Mass culling has proved effective for isolated outbreaks. However, as a result of socio-economic impacts, culling effort may vary from region to region. In the countries whose poultry systems are dominated by backyard poultries, selective culling of infected flocks is widely used because of economic concerns. For other countries, mass culling is utilized, but for large outbreaks, organizations such as Food and Agriculture Organization (FAO) and the World Health Organization (WHO) suggest a shift from wide-area culling to a modified strategy. The modified strategy entails culling of only infected and high-risk in-contact poultry along with other control measures \cite{FAO0, FAO00, Modcull, WHO0}.  
  
   In this article, we incorporated culling into a basic SI model of avian influenza. Motivated by the distinct culling strategies, we considered different functional forms for the per-capita culling rates $c_S\phi(I), c_I\psi(I)$ in the system \eqref{model1} and analyzed the dynamics.  For the general model, a sufficient condition for global stability of the disease-free equilibrium was found.  In addition, we characterized the culling rates which lead to backward bifurcation.  A more detailed analysis was conducted for three functional forms of culling rates, which modeled the distinct scenarios:  mass culling, modified culling, and selective culling.  In the case of mass culling, i.e. increasing per-capita culling rates, there is a unique globally stable endemic equilibrium when $\mathcal R_0>1$.  For modified culling, simulations show that there can be three endemic equilibria, which leads to bi-stable dynamics in the form of forward hysteresis.  Analytically, we proved that there can be one, three, or five endemic equilibria (if all equilibria are simple roots), and determined the local stability of the equilibria.  Finally, for selective culling, there can be a backward bifurcation, which also causes bi-stable dynamics.  
   
Through our exploration, we showed that non-increasing per-capita culling rates can lead to rich dynamics such as backward bifurcation and forward hysteresis, as opposed to the case of increasing per-capita culling rates.  Thus, authorities should be wary of any indication that culling effort decreases with respect to number of infected, $I$, for some values of $I$, since there may be bi-stable dynamics, as shown for modified and selective culling in our model.  In these cases, simulations suggested that temporary control measures can be employed  to ``drive the solution'' to the region of attraction corresponding to a low level equilibrium or disease free state.

In conclusion, our model and its analysis suggest that, in addition to culling, timely employment of temporary control measures such as enhanced biosecurity, isolation of poultries from wild birds and movement ban of all poultry and hatching eggs can be crucial for reducing the number of infected domestic birds to a low equilibrium level or for eliminating the disease in poultries. 

\section*{Acknowledgments} Maia Martcheva and Hayriye Gulbudak acknowledge partial support from
NSF grant DMS-0817789 and grant DMS-1220342.

\end{document}